\documentclass[prxquantum,twocolumn,notitlepage,nofootinbib,superscriptaddress,longbibliography]{revtex4-2}

\usepackage[dvipsnames]{xcolor}
\usepackage{framed}
\definecolor{shadecolor}{rgb}{0.9,0.9,0.9}

\usepackage{mathtools}
\usepackage{amsmath}
\usepackage[shortlabels]{enumitem}

\usepackage{graphicx,epic,eepic,epsfig,amsmath,latexsym,amssymb,verbatim,color}
 
\usepackage{amsfonts}       
\usepackage{nicefrac}       

\usepackage{amsmath}

\usepackage{float}
\usepackage{tikz}
\usetikzlibrary{chains}
\usetikzlibrary{fit}
\usepackage{pgflibraryarrows}		
\usepackage{pgflibrarysnakes}		

\usepackage{epsfig}
\usetikzlibrary{shapes.symbols,patterns} 
\usepackage{pgfplots}

\usepackage[strict]{changepage}
\usepackage{hyperref}
\hypersetup{colorlinks=true,citecolor=blue,linkcolor=blue,filecolor=blue,urlcolor=blue,breaklinks=true}

\usepackage[marginal]{footmisc}
\usepackage{url}
\usepackage{theorem}
 
\usepackage{theorem}
\usepackage{thmtools}
\newtheorem{definition}{Definition}
\newtheorem{proposition}{Proposition}
\newtheorem{lemma}[proposition]{Lemma}

\newtheorem{theorem}[proposition]{Theorem}
\newtheorem{corollary}[proposition]{Corollary}

\def\squareforqed{\hbox{\rlap{$\sqcap$}$\sqcup$}}
\def\qed{\ifmmode\squareforqed\else{\unskip\nobreak\hfil
\penalty50\hskip1em\null\nobreak\hfil\squareforqed
\parfillskip=0pt\finalhyphendemerits=0\endgraf}\fi}
\def\endenv{\ifmmode\;\else{\unskip\nobreak\hfil
\penalty50\hskip1em\null\nobreak\hfil\;
\parfillskip=0pt\finalhyphendemerits=0\endgraf}\fi}
\newenvironment{proof}{\noindent \textbf{{Proof~} }}{\hfill $\blacksquare$}

\newcounter{remark}
\newenvironment{remark}[1][]{\refstepcounter{remark}\par\medskip\noindent%
\textbf{Remark~\theremark #1} }{\medskip}

\newcounter{example}

\mathchardef\ordinarycolon\mathcode`\:
\mathcode`\:=\string"8000
\def\vcentcolon{\mathrel{\mathop\ordinarycolon}}
\begingroup \catcode`\:=\active
  \lowercase{\endgroup
  \let :\vcentcolon
  }

\usepackage{cleveref}
\usepackage{graphicx}
\usepackage{xcolor}

\RequirePackage[framemethod=default]{mdframed}
\newmdenv[skipabove=7pt,
skipbelow=7pt,
backgroundcolor=darkblue!15,
innerleftmargin=5pt,
innerrightmargin=5pt,
innertopmargin=5pt,
leftmargin=0cm,
rightmargin=0cm,
innerbottommargin=5pt,
linewidth=1pt]{tBox}

\newmdenv[skipabove=7pt,
skipbelow=7pt,
backgroundcolor=red!15,
innerleftmargin=5pt,
innerrightmargin=5pt,
innertopmargin=5pt,
leftmargin=0cm,
rightmargin=0cm,
innerbottommargin=5pt,
linewidth=1pt]{rBox}

\newmdenv[skipabove=7pt,
skipbelow=7pt,
backgroundcolor=blue2!25,
innerleftmargin=5pt,
innerrightmargin=5pt,
innertopmargin=5pt,
leftmargin=0cm,
rightmargin=0cm,
innerbottommargin=5pt,
linewidth=1pt]{dBox}
\newmdenv[skipabove=7pt,
skipbelow=7pt,
backgroundcolor=darkkblue!15,
innerleftmargin=5pt,
innerrightmargin=5pt,
innertopmargin=5pt,
leftmargin=0cm,
rightmargin=0cm,
innerbottommargin=5pt,
linewidth=1pt]{sBox}
\definecolor{darkblue}{RGB}{0,76,156}
\definecolor{darkkblue}{RGB}{0,0,153}
\definecolor{blue2}{RGB}{102,178,255}
\definecolor{darkred}{RGB}{195,0,0}

\newcommand{\nc}{\newcommand}
\nc{\rnc}{\renewcommand}
\nc{\lbar}[1]{\overline{#1}}
\nc{\bra}[1]{\langle#1|}
\nc{\ket}[1]{|#1\rangle}
\nc{\ketbra}[2]{|#1\rangle\!\langle#2|}
\nc{\braket}[2]{\langle#1|#2\rangle}

\nc{\proj}[1]{| #1\rangle\!\langle #1 |}
\nc{\avg}[1]{\langle#1\rangle}
\nc{\smfrac}[2]{\mbox{$\frac{#1}{#2}$}}
\nc{\tr}{\operatorname{Tr}}
\nc{\ox}{\otimes}
\nc{\dg}{\dagger}
\nc{\dn}{\downarrow}
\nc{\cA}{{\cal A}}
\nc{\cB}{{\cal B}}
\nc{\cC}{{\cal C}}
\nc{\cD}{{\cal D}}
\nc{\cE}{{\cal E}}
\nc{\cF}{{\cal F}}
\nc{\cG}{{\cal G}}
\nc{\cH}{{\cal H}}
\nc{\cI}{{\cal I}}
\nc{\cJ}{{\cal J}}
\nc{\cK}{{\cal K}}
\nc{\cL}{{\cal L}}
\nc{\cM}{{\cal M}}
\nc{\cN}{{\cal N}}
\nc{\cO}{{\cal O}}
\nc{\cP}{{\cal P}}
\nc{\cQ}{{\cal Q}}
\nc{\cR}{{\cal R}}
\nc{\cS}{{\cal S}}
\nc{\cT}{{\cal T}}
\nc{\cU}{{\cal U}}
\nc{\cV}{{\cal V}}
\nc{\cX}{{\cal X}}
\nc{\cY}{{\cal Y}}
\nc{\cZ}{{\cal Z}}
\nc{\cW}{{\cal W}}
\nc{\csupp}{{\operatorname{csupp}}}
\nc{\qsupp}{{\operatorname{qsupp}}}
\nc{\var}{{\operatorname{var}}}
\nc{\rar}{\rightarrow}
\nc{\lrar}{\longrightarrow}
\nc{\polylog}{{\operatorname{polylog}}}
\nc{\wt}{{\operatorname{wt}}}
\nc{\av}[1]{{\left\langle {#1} \right\rangle}}
\nc{\supp}{{\operatorname{supp}}}

\nc{\argmin}{{\operatorname{argmin}}}

\def\x{\xi}

\nc{\RR}{{{\mathbb R}}}
\nc{\CC}{{{\mathbb C}}}
\nc{\FF}{{{\mathbb F}}}
\nc{\NN}{{{\mathbb N}}}
\nc{\ZZ}{{{\mathbb Z}}}
\nc{\PP}{{{\mathbb P}}}
\nc{\QQ}{{{\mathbb Q}}}
\nc{\UU}{{{\mathbb U}}}
\nc{\EE}{{{\mathbb E}}}
\nc{\id}{{\operatorname{id}}}

\nc{\CHSH}{{\operatorname{CHSH}}}

\nc{\be}{\begin{equation}}
\nc{\ee}{{\end{equation}}}
\nc{\bea}{\begin{eqnarray}}
\nc{\eea}{\end{eqnarray}}
\nc{\<}{\langle}
\rnc{\>}{\rangle}
\nc{\rU}{\mbox{U}}

\nc{\ob}[1]{#1}

\nc{\SEP}{{\text{\rm SEP}}}
\nc{\NS}{{\text{\rm NS}}}
\nc{\LOCC}{{\text{\rm LOCC}}}
\nc{\PPT}{{\text{\rm PPT}}}
\nc{\EXT}{{\text{\rm EXT}}}
\nc{\Sym}{{\operatorname{Sym}}}

\nc{\ERLO}{{E_{\text{r,LO}}}}
\nc{\ERLOCC}{{E_{\text{r,LOCC}}}}
\nc{\ERPPT}{{E_{\text{r,PPT}}}}
\nc{\ERLOCCinfty}{{E^{\infty}_{\text{r,LOCC}}}}
\nc{\Aram}{{\operatorname{\sf A}}}

\usepackage{tikz}

\makeatletter
\def\grd@save@target#1{%
  \def\grd@target{#1}}
\def\grd@save@start#1{%
  \def\grd@start{#1}}
\tikzset{
  grid with coordinates/.style={
    to path={%
      \pgfextra{%
        \edef\grd@@target{(\tikztotarget)}%
        \tikz@scan@one@point\grd@save@target\grd@@target\relax
        \edef\grd@@start{(\tikztostart)}%
        \tikz@scan@one@point\grd@save@start\grd@@start\relax
        \draw[minor help lines,magenta] (\tikztostart) grid (\tikztotarget);
        \draw[major help lines] (\tikztostart) grid (\tikztotarget);
        \grd@start
        \pgfmathsetmacro{\grd@xa}{\the\pgf@x/1cm}
        \pgfmathsetmacro{\grd@ya}{\the\pgf@y/1cm}
        \grd@target
        \pgfmathsetmacro{\grd@xb}{\the\pgf@x/1cm}
        \pgfmathsetmacro{\grd@yb}{\the\pgf@y/1cm}
        \pgfmathsetmacro{\grd@xc}{\grd@xa + \pgfkeysvalueof{/tikz/grid with coordinates/major step}}
        \pgfmathsetmacro{\grd@yc}{\grd@ya + \pgfkeysvalueof{/tikz/grid with coordinates/major step}}
        \foreach \x in {\grd@xa,\grd@xc,...,\grd@xb}
        \node[anchor=north] at (\x,\grd@ya) {\pgfmathprintnumber{\x}};
        \foreach \y in {\grd@ya,\grd@yc,...,\grd@yb}
        \node[anchor=east] at (\grd@xa,\y) {\pgfmathprintnumber{\y}};
      }
    }
  },
  minor help lines/.style={
    help lines,
    step=\pgfkeysvalueof{/tikz/grid with coordinates/minor step}
  },
  major help lines/.style={
    help lines,
    line width=\pgfkeysvalueof{/tikz/grid with coordinates/major line width},
    step=\pgfkeysvalueof{/tikz/grid with coordinates/major step}
  },
  grid with coordinates/.cd,
  minor step/.initial=.2,
  major step/.initial=1,
  major line width/.initial=2pt,
}
\makeatother

\usepackage{thmtools}
\usepackage{thm-restate}
\usepackage{etoolbox}
\makeatletter
\def\problem@s{}
\newcounter{problems@cnt}

\newcommand{\allproblems}{\problem@s}
\makeatother

\usepackage{amsmath,amsfonts,amssymb,array,graphicx,mathtools,multirow,bm,times,tcolorbox,relsize,booktabs}
\usepackage[utf8]{inputenc}
\usepackage[T1]{fontenc}
\usepackage{tikz} 
\usetikzlibrary{quantikz2}
\usepackage[ruled, vlined, linesnumbered]{algorithm2e}

\usepackage[vcentermath, enableskew]{youngtab}
\Yboxdim{5pt}

\definecolor{colortwo}{rgb}{0.4,0.77,0.17}
\definecolor{colorthree}{rgb}{0.01,0.51,0.93}

\allowdisplaybreaks

\begin{document}
\title{LCQNN:~Linear Combination of Quantum Neural Networks}

\author{Hongshun Yao}
\email{yaohongshun2021@gmail.com}
\affiliation{Thrust of Artificial Intelligence, Information Hub,\\
Hong Kong University of Science and Technology (Guangzhou), Guangdong 511453, China}
\author{Xia Liu}
\email{liuxia1113@outlook.com}
\affiliation{Thrust of Artificial Intelligence, Information Hub,\\
Hong Kong University of Science and Technology (Guangzhou), Guangdong 511453, China}
\author{Mingrui Jing}
\email{mjing638@connect.hkust-gz.edu.cn}
\affiliation{Thrust of Artificial Intelligence, Information Hub,\\
Hong Kong University of Science and Technology (Guangzhou), Guangdong 511453, China}
\author{Guangxi Li}
\email{liguangxi@quantumsc.cn}
\affiliation{Quantum Science Center of Guangdong-Hong Kong-Macao Greater Bay Area, Shenzhen 518045, China}
\affiliation{Thrust of Artificial Intelligence, Information Hub,\\
Hong Kong University of Science and Technology (Guangzhou), Guangdong 511453, China}
\author{Xin Wang}
\email{felixxinwang@hkust-gz.edu.cn}
\affiliation{Thrust of Artificial Intelligence, Information Hub,\\
Hong Kong University of Science and Technology (Guangzhou), Guangdong 511453, China}

\begin{abstract}
Quantum neural networks combine quantum computing with advanced data-driven methods, offering promising applications in quantum machine learning. However, the optimal paradigm for balancing trainability and expressivity in QNNs remains an open question. To address this, we introduce the Linear Combination of Quantum Neural Networks (LCQNN) framework, which uses the linear combination of unitaries concept to create a tunable design that mitigates vanishing gradients without incurring excessive classical simulability. We show how specific structural choices, such as adopting $k$-local control unitaries or restricting the model to certain group-theoretic subspaces, prevent gradients from collapsing while maintaining sufficient expressivity for complex tasks. We further employ the LCQNN model to handle supervised learning tasks, demonstrating its effectiveness on real datasets. In group action scenarios, we show that by exploiting symmetry and excluding exponentially large irreducible subspaces, the model circumvents barren plateaus. Overall, LCQNN provides a novel framework for focusing quantum resources into architectures that are practically trainable yet expressive enough to tackle challenging machine learning applications.
\end{abstract}

\date{\today}
\maketitle


\section{Introduction}
Quantum machine learning (QML) has quickly gained traction as a hotspot research direction for integrating quantum computing with data-driven and machine learning methodologies~\cite{schuld2015introduction,biamonte2017quantum,zhang2020toward,li2022recent,qian2022dilemma,caro2022generalization,cerezo2022challenges,tian2023recent,jerbi2023quantum,du2025quantum}. Within this broad area, quantum neural networks (QNNs)~\cite{benedetti2019parameterized,ostaszewski2021structure} have been proposed to harness the computational resources of quantum devices in tandem with ideas drawn from deep-learning architectures, demonstrating potential in areas such as image recognition, precise quantum chemistry applications and optimization algorithms~\cite{bauer2020quantum,hempel2018quantum,farhi2014quantum,zhou2020quantum,chen2020variational,takaki2021learning,zhang2024statistical}.

Nevertheless, effectively training QNNs encounters several ongoing challenges, including Barren Plateaus (BP) impeding gradient-based optimization by creating extremely flat regions in the loss landscape, and local minima trapping models away from global optima~\cite{Mcclean2018barren,Cerezo2021cost,You2021exponentially,zhang2024statistical}. To address these limitations, solutions such as circuit architecture refinements, parameter initialization heuristics, and error-mitigation techniques~\cite{Liu2024mitigating,park2024hamiltonian,wiersema2020exploring,farhi2014quantum,zhou2020quantum,Zhang2022,wang2024trainability,shi2024avoiding} have been developed to navigate hardware constraints and convergence issues when developing powerful quantum-enhanced learning models. Recently, unified theories of BP for trace-form loss function have been demonstrated~\cite{Fontana2024characterizing,Ragone2024lie}, which have accelerated the investigation on different QNNs and their trainability with respect to specific problem structures~\cite{Nguyen2024theory,schatzki2024theoretical,Allcock2024dynamical,kazi2024analyzing,Mao2024towards,Jing2025quantum}.

Although those efforts have been made to eliminate the fundamental scalability issue of QNNs, a recent perspective article raises a significant point that many circuit design strategies avoiding BP, however, contain structural constraints that sometimes enable efficient classical simulability, limiting the potential for genuine quantum advantage~\cite{Cerezo2023does}. One of the core insights is that introducing symmetries or limiting circuit depth can keep gradients from vanishing, but at the risk of reducing the effective dimensionality of the parameter space to a region that classical methods handle efficiently. Consequently, while shallow or highly structured QNNs appear more trainable, researchers have questioned whether these BP-free architectures truly leverage quantum resources or can be replaced with classical procedures, especially when an initial data-gathering phase on a quantum device suffices to replicate the loss landscape on classical hardware~\cite{Yu2022power,Ragone2024lie}.

Two prominent questions arise. First, \textit{to what extent can design choices for QNNs be inspired by powerful quantum algorithms that achieve proven advantages in areas such as solving linear equations~\cite{Harrow2009quantum,Childs2012hamiltonian}, function approximation and state searching~\cite{Rossi2022multivariable,Low2017optimal,Gilyen2019quantum}.} The rationale is that methods used by well-known quantum algorithms to navigate large state spaces might help QNNs avoid barren plateaus without simultaneously falling back into classical simulability. Second, \textit{how should we formally examine both trainability and expressivity on these new structured QNNs, ensuring that the QNN remains not only free from exponentially vanishing gradients but also nontrivial from a computational perspective.}

To address these, we design a new QNN framework that integrates the concepts of linear combination of unitaries (LCU)~\cite{Childs2012hamiltonian} with tunable circuits and combination coefficients. We begin by introducing the circuit construction of the so-called \textit{linear combination of quantum neural networks} (LCQNN), which consists of a coefficient parametrization layer and a control-unitary layer. Next, we provide detailed demonstrations on the trainability and expressivity of LCQNN by inserting QNNs into the control unitary layer. Our results showcase theoretical scaling on the gradient variance for different control QNNs, quantitatively connecting the trainability with system size and the number of control qubits. Furthermore, the combination of $k$-local systems is investigated, demonstrating that the system dimension $k$ predominates the variance of the cost function. Finally, we extend our investigation towards group action, showcasing that our LCQNN can mitigate or even avoid BP, considering structured observables.Our contributions are multi-fold:
\vspace{-1ex}
\begin{itemize}
    \item \textbf{Framework Design}: We propose LCQNN, a novel QNN framework that systematically balances expressivity and trainability. It is the first to use a learnable linear combination of parameterized unitaries to explicitly manage this trade-off, moving beyond ad-hoc architectural heuristics.
    \item \textbf{Theoretical Guarantees}: 
    \begin{enumerate}
        \item We provide a rigorous theoretical analysis of trainability for the LCQNN framework. We prove that its cost function's gradient variance scales favorably, inversely with the number of combined QNNs ($L$).
        \item We demonstrate the versatility of LCQNN by restricting it to the $k$-local scenario, showing that the scale of the local QNN, determines the trainability of the LCQNN framework. This proposition implies that LCQNN's structure can be tailored to specific problems to guarantee trainability while retaining non-trivial quantum features.
        \item We generate the LCQNN framework for the group action setting, offering novel perspectives for the design of circuit structures within the geometric quantum machine learning paradigm.
    \end{enumerate}   
    \item \textbf{Practical Validation}: We demonstrate the correctness of the theoretical framework through gradient analysis experiments and the model's potential for practical application by testing it on real-world datasets.
\end{itemize}


\section{Method}\label{sec:lcqnn}
\subsection{Preliminaries and Notations}
In quantum computing, the basic unit of quantum information is a quantum bit or qubit. A single-qubit pure state is described by a unit vector in the Hilbert space $\mathbb{C}^2$, which is commonly written in Dirac notation $\ket{\psi} = \alpha\ket{0} + \beta\ket{1}$, with $\ket{0} = (1,0)^T$, $\ket{1} =(0,1)^T$, $\alpha,\beta\in \mathbb{C}$ subject to $|\alpha|^2 + |\beta|^2 = 1$. The complex conjugate of $\ket{\psi}$ is denoted as $\bra{\psi} = \ket{\psi}^\dagger$. The Hilbert space of $N$ qubits is formed by the tensor product ``$\otimes$'' of $N$ single-qubit spaces with dimension $d=2^N$. General mixed quantum states are represented by the density matrix, which is a positive semidefinite matrix $\rho \in \mathbb{C}^{d\times d}$ subject to $\tr[\rho]=1$. 

Quantum gates are unitary matrices, which transform quantum states via matrix-vector multiplication. Common single-qubit rotation gates include $R_x(\theta)=e^{-i\theta X/2}$, $R_y(\theta)=e^{-i\theta Y/2}$, $R_z(\theta)=e^{-i\theta Z/2}$, which are in the matrix exponential form of Pauli matrices,
\begin{equation}
    X = \begin{pmatrix}
        0 & 1 \\ 1 & 0
    \end{pmatrix},\quad\quad
    Y = \begin{pmatrix}
        0 & -i \\ i & 0 \\
    \end{pmatrix},\quad\quad
    Z = \begin{pmatrix}
        1 & 0 \\ 0 & -1 \\
    \end{pmatrix}.
\end{equation}
Common two-qubit gates include controlled-X (or CNOT) gate $C\text{-}X=I\oplus X$ ($\oplus$ is the direct sum), which can generate quantum entanglement among qubits.

\subsection{LCQNN Framework}\label{sec:lcqnn_framework}

Quantum neural networks (QNNs) require balancing two competing properties: expressivity and trainability. Expressivity refers to the range of quantum states that a QNN can achieve. A QNN with high expressivity can generate a wide variety of quantum states, a prerequisite for solving complex problems. Trainability, on the other hand, measures how effectively a QNN can be trained using classical algorithms, such as gradient-based methods. A trainable QNN exhibits non-vanishing gradients, enabling efficient convergence to optimal parameters. While deep QNNs are typically favored for their strong expressivity, prior research has revealed a fundamental trade-off: strong expressivity will lead to poor trainability due to the BP. In this context, it is difficult to find a flexible method for designing QNNs that strikes a balance between the expressivity and the trainability of QNN to efficiently complete the task. 

Motivated by the trade-off between trainability and expressivity, and the requirements for task-based structural design in QNNs, we propose a novel framework termed linear combination of quantum neural networks (LCQNN), which is inspired by the linear combination of unitaries (LCU) technique. The LCQNN architecture operates on an $(m+n)$-qubit system and comprises two parts: \textit{coefficient layer} , denoted as $V(\alpha)$, and a \textit{control unitary layer} referred to as ${\rm C}\text{-}U$. The coefficient layer $V(\alpha)$ acts on the first $m$-qubit of LCQNN, while the unitaries in the control unitary layer are applied to the last $n$-qubit. The whole structure of LCQNN is depicted at the bottom of Fig.~\ref{fig:lcqnn} and can be expressed as 
\begin{equation}
    W(\bm \alpha,\bm \theta)=\prod_{j=0}^{L-1}\left({\rm C}\text{-}U_j\right)\cdot(V\otimes I^{\otimes n}), \;(L\in\mathbb{R}_+),
\end{equation}
where ${\bm\alpha}=\{\alpha^0_1,\alpha^1_{1},\alpha^1_{2},\cdots,\alpha^{m-1}_1,\alpha^{m-1}_{2^{m-1}}\}$ and ${\bm \theta}=\{\theta_1,\cdots,\theta_{2^m}\}$ are two parameter sets. 
We can tune parameters in each layer and adjust the number of unitary $U$ to govern the expressivity and trainability. Note that the coefficient layer can be constructed by rotation gates, as shown at the top of Fig.~\ref{fig:lcqnn}. Additionally, if the control-unitary layer is substituted with control-permutation operators, this model degenerate into the QSF framework that is a unified framework for achieving nonlinear functions of quantum states~\cite{yao2024nonlinear}.

In the following, we conduct a more detailed analysis of the trainability and expressivity of LCQNN in terms of the gradient variance and state generation, respectively.

\begin{figure}[t]
    \centering
    \includegraphics[width=\linewidth]{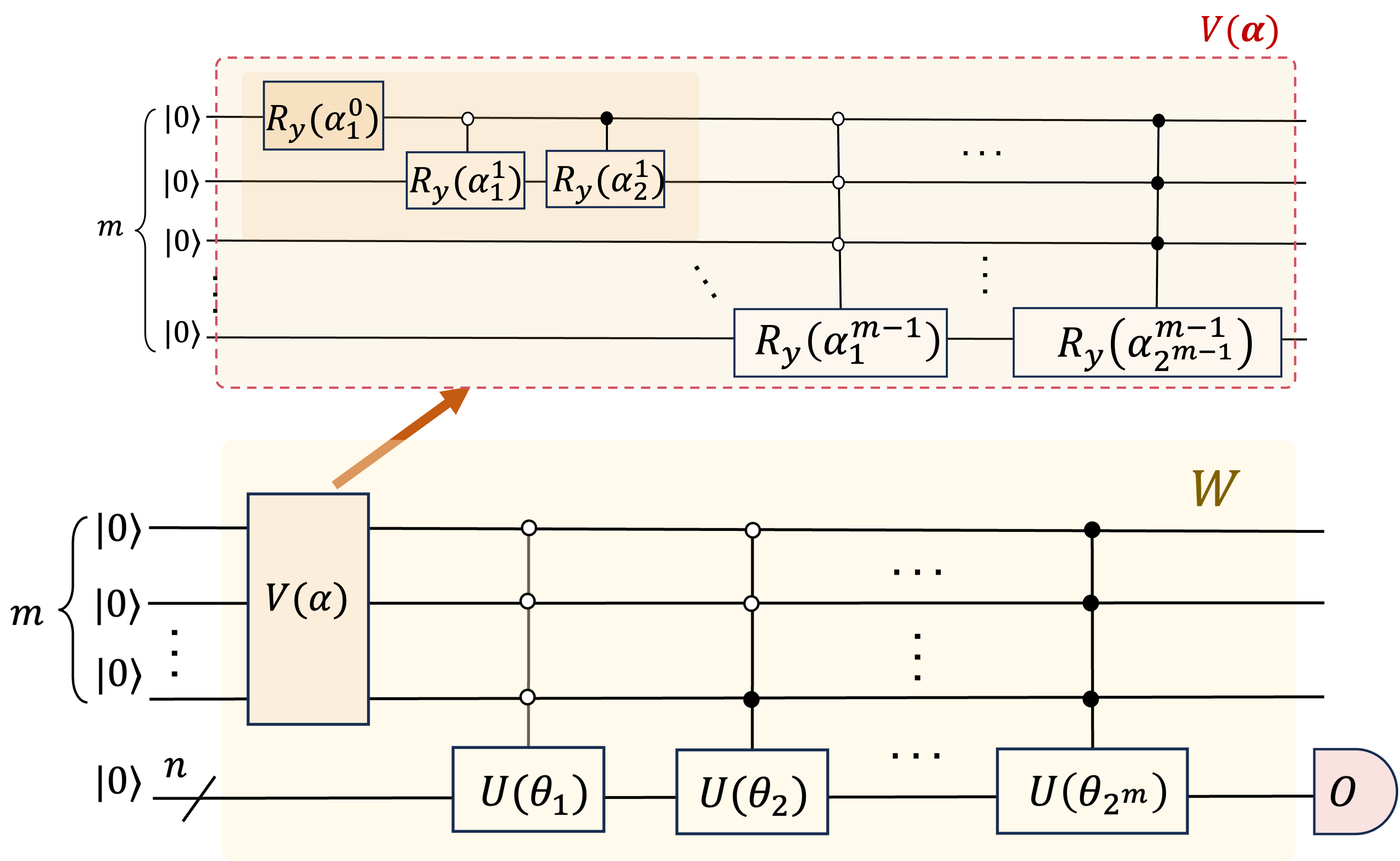}
    \caption{\textbf{Illustration of LCQNN.} $V(\bm\alpha)$ is the coefficient layer, which can be constructed by the $R_y$ rotation gate. The control-unitary layer consists of at most $2^m$ unitary for $m$ control qubits.}
    \label{fig:lcqnn}
\end{figure}

\subsection{Trainability and Expressivity of LCQNN}\label{sec:lcqnn_trainability_1}
In this section, we analyze the expressivity of LCQNN in generating an arbitrary $m+n$-qubit pure state, while also evaluating its trainability by calculating the variance of the cost function, given that the $U(\theta_j)$ in LCQNN is sampled from unitary $2$-design ensembles, and $V(\alpha)$ is constructed by a $R_y$ rotation gates.

Trainability is one of the most critical concerns in QNNs. When a QNN forms a unitary 2-design, the cost gradients vanish exponentially with system size, severely hampering optimization.  
To formalize the trainability of LCQNN, we consider the cost function $C(\bm\alpha,\bm\theta)$ as the expectation value of an observable $O$ with respect to an initial state $\ket{0}^{\otimes (m+n)}$, and a LCQNN $W$ on $m+n$ qubits. We examine the gradient component $\partial C(\bm\alpha,\bm\theta)$ with respect to the variational parameter $\bm\alpha$ and $\bm\theta$. Here, we assume that the unitary $U(\theta_i)$ in the control-unitary layer are sampled from local unitary $2$-designs. Under this assumption, the trainability of LCQNN can be improved by reducing the number of unitaries in the control-unitary layer. Theorem~\ref{thm:trade-off_exp_bp} below demonstrates the expressivity of LCQNN in generating any $m+n$-qubit pure states, and establishes an explicit relationship between its trainability and the number of unitary in the control-unitary layer. 

\begin{theorem}[Trainability and expressivity]\label{thm:trade-off_exp_bp}
    Denote $L\in\mathbb{R}_+$ and $W:=\prod_{j=0}^{L-1}{\rm C}\text{-}U_j\cdot(V\otimes I^{\otimes n})$, which is described in Fig.\ref{fig:lcqnn}. Then, the LCQNN $W$ can achieve any pure state $\ket{\psi}:=\oplus_{j=0}^{2^m-1}\sqrt{p_j}\ket{\psi_j}\in(\mathbb{C}^2)^{\otimes m+n}$, where $p_j=0$ for $j\geq L$. For given any traceless observable $O$ on system $(\mathbb{C}^2)^{\otimes n}$, we have $\mathbb{E}_{\bm \alpha,\bm \theta}[\partial C(\bm \alpha,\bm \theta)]=0$ and,
    \begin{equation}
        \begin{aligned}
            {\rm Var}_{\bm \alpha,\bm \theta}[\partial C(\bm \alpha,\bm \theta)]\in\cO\left(\frac{{\rm Poly}(n)}{L2^{n}}\right),
        \end{aligned}
    \end{equation}
    where $C(\bm \alpha,\bm \theta):=\bra{0}^{\otimes m+n}W(I\otimes O)W^\dagger\ket{0}^{\otimes m+n}$ denotes the cost function.  
\end{theorem}
\begin{proof}
        Firstly, we are going to demonstrate the expressivity. Circuit $V$ aims to generate a superposition state based on the coefficients $\{p_j\}$, a probability distribution satisfying $\sum_{j=0}^{2^t-1} p_j=1$, where denotes $L=2^t$ (w.l.o.g.). In specific, after applying the constructed circuit $V(\bm\alpha)$ shown in Fig.~\ref{fig:lcqnn}, we can obtain the target state:
    \begin{equation}
        \begin{aligned}
            V(\bm\alpha)\ket{0}^{\otimes m}=\sum_{j=0}^{L-1}\sqrt{p_j(\bm\alpha)}\ket{j},
        \end{aligned}
    \end{equation}
    where $p_j(\bm\alpha):=\prod_{k=1}^t[\cos^2(\alpha_{j_k})]^{1-j_k}[\sin^2(\alpha_{j_k})]^{j_k}$, $j=\sum_{k=1}^t j_k 2^{t-k}$, and trainable parameters $\alpha_{j_k}\in\bm\alpha:=\{\alpha_0,\cdots,\alpha_{L-2}\}$, and $\ket{j}$ is the basis on the Hilbert space $(\mathbb{C}^2)^{\otimes m}$. As a result, after applying the whole circuit $W$, one can obtain the output state as follows:
    \begin{equation}\label{eq:pqc_state}
        \begin{aligned}
            W\ket{0}^{\otimes m+n}&=\prod_{j=0}^{2^t-1}C\text{-}U_j(\bm \theta_j)\cdot\bigoplus_{j=0}^{2^t-1}\sqrt{p_j(\bm\alpha)}\ket{0}^{\otimes n}\\
            &=\bigoplus_{j=0}^{2^t-1}\sqrt{p_j(\bm\alpha)}U_j(\bm \theta_j)\ket{0}^{\otimes n},
        \end{aligned}
    \end{equation}
    where $\bm \theta_j\in \bm \theta=\{\bm \theta_0,\cdots,\bm \theta_{L-1}\}$. Notice that $U_j$ is universal with some parameter $\bm\theta_j$. Thus, there exist parameters $\bm\theta$ and $\bm\alpha$, such that the state in Eq.~\eqref{eq:pqc_state} can achieve any target pure state $\ket{\psi}$ as we assumed, which demonstrates its expressivity.
    
    Secondly, we turn to investigate the trainability by focusing on these two terms $\mathbb{E}_{\bm \alpha,\bm \theta}[\partial C(\bm \alpha,\bm \theta)]$ and $\text{Var}_{\bm \alpha,\bm \theta}[\partial C(\bm \alpha,\bm \theta)]$. We further derive the cost function, which can be written in the following form:
    \begin{equation}\label{eq:loss_function}
        \begin{aligned}
            C(\bm\alpha,\bm\theta):=\sum_{j=0}^{2^t-1}p_j(\bm\alpha)\bra{0}^{\otimes n}U^\dagger_j(\bm\theta_j)OU(\bm\theta_j)\ket{0}^{\otimes n}.
        \end{aligned}
    \end{equation}
    Notice that the partial derivative with respect to $\alpha_j$ satisfies $\partial_{\alpha_j} C(\bm \alpha,\bm \theta)=\partial_{\alpha_j}p_j(\bm\alpha)\bra{0}^{\otimes n}U^\dagger_j(\bm\theta_j)OU(\bm\theta_j)\ket{0}^{\otimes n}$, which means $\mathbb{E}_{\bm \alpha,\bm \theta}[\partial_{\alpha_j} C(\bm \alpha,\bm \theta)]=0$. Similarly, we have $\mathbb{E}_{\bm \alpha,\bm \theta}[\partial_{\theta_j} C(\bm \alpha,\bm \theta)]=0$, where $\theta_j\in\mathbb{R}$ denotes the trainable parameter, an element in the vector $\bm\theta_j\in \mathbb{R}^{L}$. Therefore, we have $\mathbb{E}_{\bm \alpha,\bm \theta}[\partial C(\bm \alpha,\bm \theta)]=0$. Next, we study this term $\mathbb{E}_{\bm \alpha,\bm \theta}[\partial^2 C(\bm \alpha,\bm \theta)]$. Specifically, we have the following observations:
    \begin{equation}
        \begin{aligned}
            \mathbb{E}_{\bm \alpha}[p^2_{j}(\bm\alpha)]\in\cO\left(\frac{1}{2^t}\right),\quad\mathbb{E}_{\bm \alpha}[\partial^2_{\alpha_j}p_{j}(\bm\alpha)]\in\cO\left(\frac{1}{2^t}\right),
        \end{aligned}
    \end{equation}
    which can be calculated by direct integration of $\bm\alpha$. Then, we consider the fact that $\bm\alpha$ and $\bm\theta$ are independent, and that $\mathbb{E}_{\bm\theta}[\partial^2_{\theta_j}\bra{0}^{\otimes n}U^\dagger_j(\bm\theta_j)OU(\bm\theta_j)\ket{0}^{\otimes n}]\in\cO(\text{Poly(n)}/2^n)$, raising from the standard BP issue. As a result, we obtain $\mathbb{E}_{\bm \alpha,\bm \theta}[\partial^2_{\theta_j} C(\bm \alpha,\bm \theta)]$, i.e.,
    \begin{equation}
        \begin{aligned}
            \mathbb{E}_{\bm\alpha}[p^2_{j}(\bm\alpha)]\cdot\mathbb{E}_{\bm\theta}[\partial^2_{\theta_j}\bra{0}^{\otimes n}U^\dagger_j(\bm\theta_j)OU(\bm\theta_j)\ket{0}^{\otimes n}].
        \end{aligned}
    \end{equation}
    Similarly, we also have $ \mathbb{E}_{\bm \alpha,\bm \theta}[\partial^2_{\alpha_j} C(\bm \alpha,\bm \theta)]\in\cO(\text{Poly}(n)/(L2^n))$. It means that $\text{Var}_{\bm \alpha,\bm \theta}[\partial C(\bm \alpha,\bm \theta)]\in\cO(\text{Poly}(n)/(L2^{n)})$, which completes this proof.
\end{proof}

\begin{remark}\label{thm:fully_contrl}
   If $L=2^{m}$, the parameted circuit $W$ can achieve any pure state on the system $(\mathbb{C}^2)^{\otimes m+n}$ demonstrating its expressivity, and the trainability suffers from the BP issue over the whole system, i.e., $\text{Var}_{\bm \alpha,\bm \theta}[\partial C(\bm \alpha,\bm \theta)]\in\cO\left(\frac{\text{Poly}(n)}{2^{(m+n)}}\right)$. 
\end{remark}

Theorem~\ref{thm:trade-off_exp_bp} establishes an explicit trade-off between expressivity and trainability in quantum neural networks (QNNs) through carefully designed circuit architectures. Notably, this quantified relationship may facilitate applications in quantum architecture search (QAS)~\cite{martyniuk2024quantum,zhang2024predicting,zhu2025scalable}.

\begin{figure}[t]
    \centering
    \includegraphics[width=\linewidth]{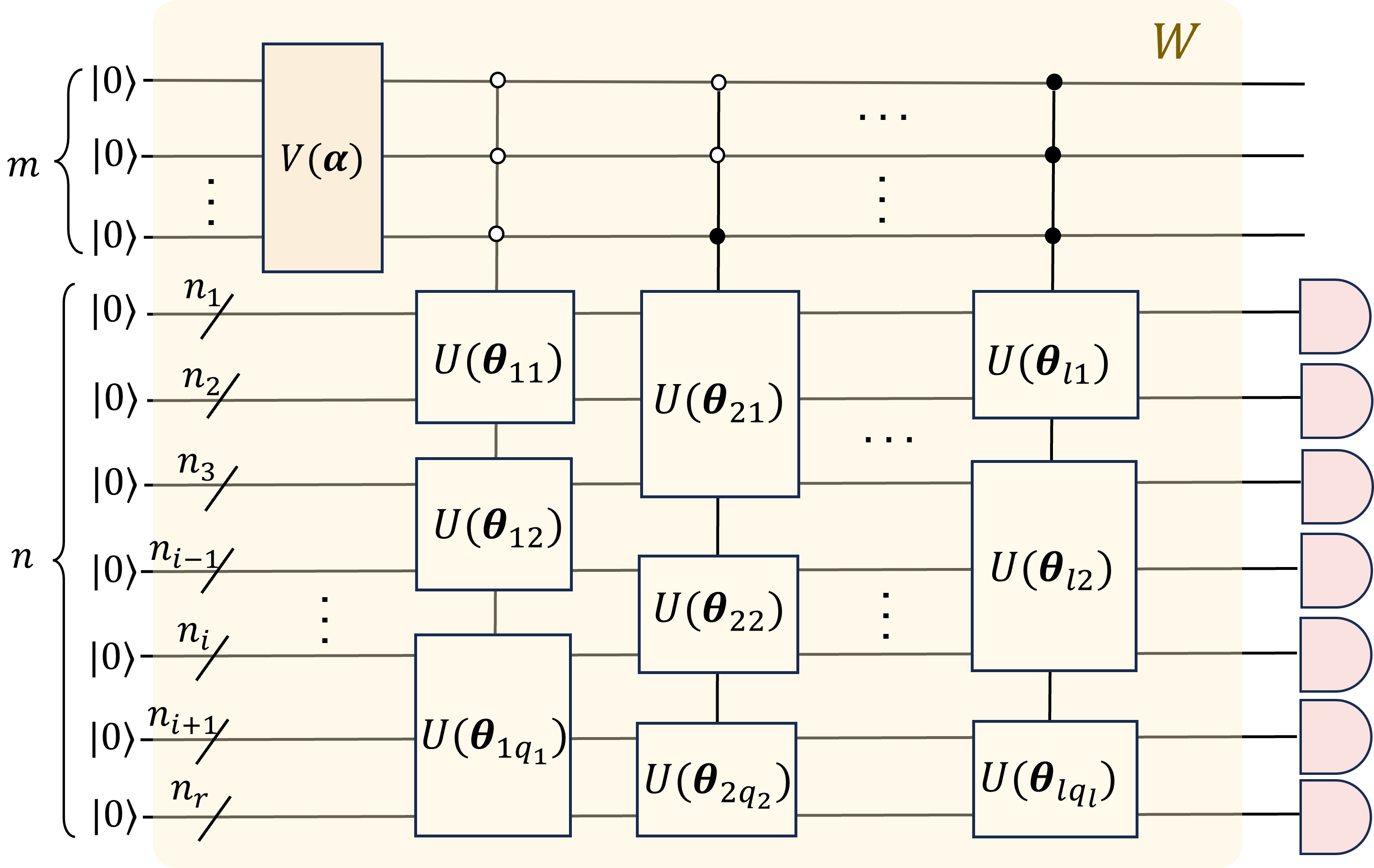}
    \caption{\textbf{LCQNN with $k$-local unitary.} The $U(\bm\theta_{jt})$ in the control-unitary layer acts on $k$-local systems and measurement operators act on $n_i$ systems, respectively. The coefficient layer $V(\bm\alpha)$ can be constructed by a control $R_y$ gate.}
    \label{fig:k-local}
\end{figure}

\subsection{LCQNN with $k$-local Unitary}\label{sec:lcqnn_trainability_2}
As illustrated in Fig.~\ref{fig:lcqnn}, the assumption of global universality over the entire $n$-qubit subspace $(\mathbb{C}^2)^{\otimes n}$ grants quantum neural networks (QNNs) excessive expressivity, which fundamentally contributes to barren plateau (BP) issues. To mitigate this phenomenon, we propose an architecture employing $k$-local unitaries applied to individual subsystems rather than to the full $n$-qubit system. The resulting $k$-local LCQNN structure is depicted in Fig.~\ref{fig:k-local}, yielding the following key property.

\begin{proposition}\label{cor:k-local}
     The LCQNN with $k$-local unitary is defined as $W:=\prod_{j=0}^{L-1}\left[{\rm C}\text{-}\left(\bigotimes_tU_{jt}\right)\right]\cdot(V\otimes I^{\otimes n})$, where $U_{jt}$ is a $k$- local unitary $2$-design, as shown in Fig.~\ref{fig:k-local}.
    Then, the LCQNN $W$ can achieve the pure state $\ket{\psi}:=\bigoplus_{j=0}^{2^m-1}\left(\sqrt{p_j}\ket{\psi_j}\bigotimes_i\ket{\phi_{jt}}\right)\in(\mathbb{C}^2)^{\otimes m+n}$, with $p_j=0$ for $j\geq L$. Furthermore, the variance of the cost gradient scales with the number of qubits as 
    \begin{equation}
        \begin{aligned}
            {\rm Var}_{\bm \alpha,\bm \theta}[\partial C(\bm \alpha,\bm \theta)]\in\cO\left(\frac{{\rm Poly}(k)}{L2^{k}}\right).
        \end{aligned}
    \end{equation}
\end{proposition}

The detailed proof can be found in Appendix~\ref{appendix:detailed_proof}. The variance of the cost function comprises two distinct contributions: the \textit{coefficient layer} $V(\bm\alpha)$ and the \textit{control-unitary layer} $\{{\rm C}\text{-}U_{jt}\}$. For the control-unitary layer, Proposition~\ref{cor:k-local} demonstrates that employing local unitaries and appropriate observable mitigates barren plateau issues at the cost of reduced expressivity.

It is important to note that, within the context of classical machine learning, each measurement operator $O_{\bm\theta}:=U^\dagger(\bm\theta)OU(\bm\theta)$ can be regarded as a feature extraction process that retrieves information from the quantum state $\rho$. As illustrated in Fig.~\ref{fig:framework}, the traditional QNNs correspond to a robust feature extraction process (universal $U(\bm\theta)$). However, due to the constraints imposed by the BP issue, the stronger the feature extraction capability, the more difficult it becomes to train the QNNs, leading to scalability challenges for traditional QNN models. Inspired by the progressive feature extraction mechanisms observed in classical machine learning models, the proposed LCQNN in this paper achieves specific machine learning tasks by integrating relatively weaker feature extractors ($k$-local unitary $\otimes_tU(\bm\theta_{jt})$) while ensuring the trainability of the model. This approach effectively balances expressivity and trainability, thereby identifying a task-based QNN architecture.

\begin{figure*}[t]
    \centering
    \includegraphics[width=0.7\linewidth]{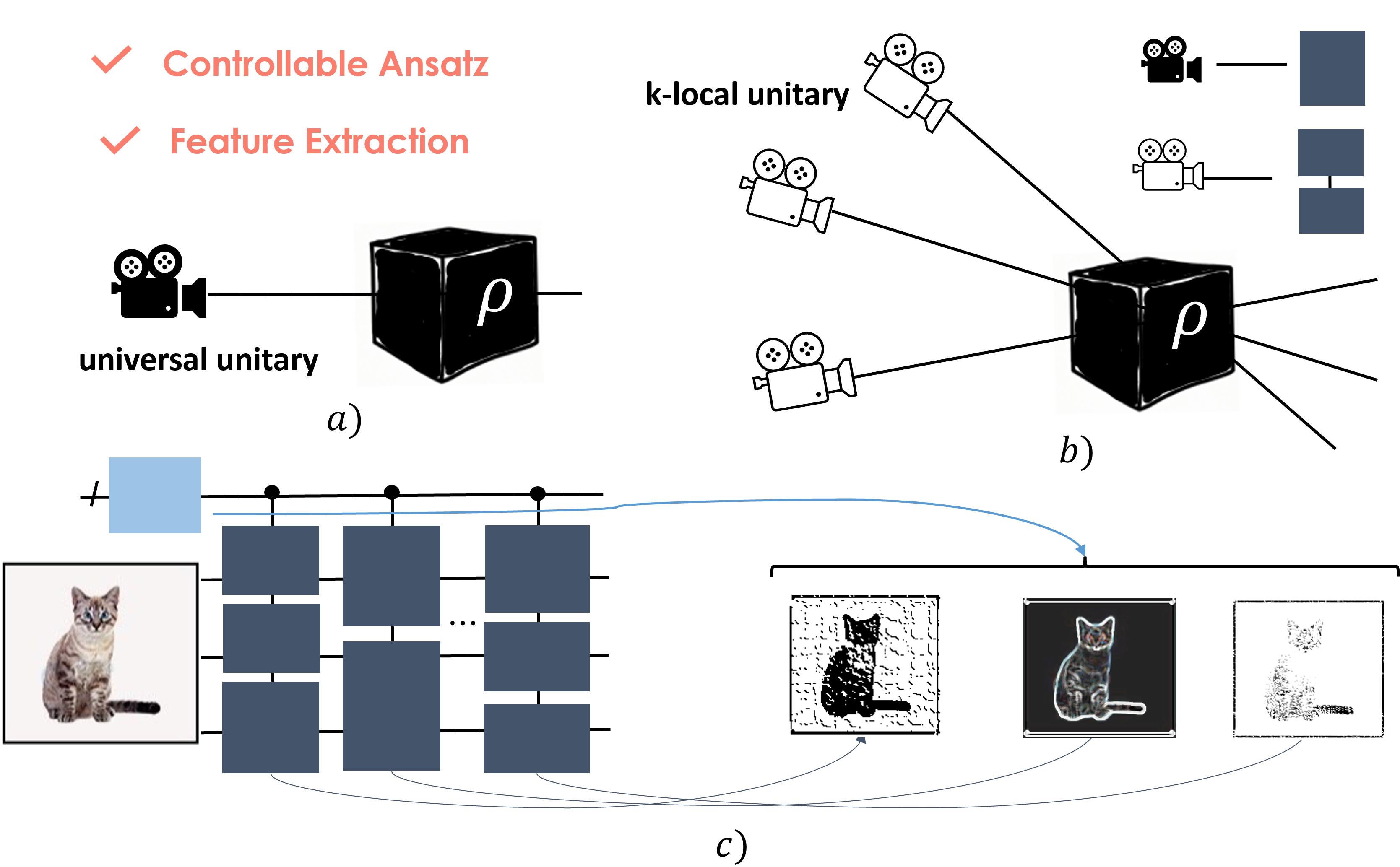}
    \caption{Comparison between LCQNN and traditional QNNs. a) and b) denote the conventional QNN and the LCQNN, respectively. c) Illustrates the framework of feature extraction using LCQNN. Specifically, each local QNN can be regarded as a feature filter, which, through the coefficient layer, combines different features to achieve the purpose of pattern recognition. In comparison with traditional QNN models, LCQNN not only possesses the capability to accomplish tasks but also demonstrates trainability.}
    \label{fig:framework}
\end{figure*}

\subsection{LCQNN over Group Action}\label{sec:lcqnn_trainability_group}
Next, we investigate the general LCQNN, that is, a linear combination of parameterized unitary actions $\textbf{R}(G)$ of a finite group or compact Lie group $G$.

Recall that the unitary representation $\textbf{R}(G)$ can be rewritten as a direct sum of irreducible representations (\textit{irreps}) by the isotypical decomposition. Denote the Hilbert space as $\cH:=(\mathbb{C}^2)^{\otimes N}$, we have
\begin{equation}
    \begin{aligned}
        \cH\cong\bigoplus_{\mu\in \hat{G}}\cQ_{\mu}\otimes\mathbb{C}^{m_\mu},
    \end{aligned}
\end{equation}
where denotes the set of labels of irreducible representations appearing in $\textbf{R}$ by $\hat{G}$, $\cQ_\mu$ is the irreducible sub-space, and $m_\mu$ expresses the multiplicity of the irrep $\cQ_\mu$. Here, we denote the dimensions of $\cQ_\mu$ as $d_\mu$. The change of basis refers to $U_{\text{transform}}$.

Suppose the observable $O$ can be written as a block-diagonal form under the basis transform, i.e.,
\begin{equation}\label{eq:basis_transform}
    \begin{aligned}
        U_{\text{transform}}^\dagger OU_{\text{transform}}=\bigoplus_{\mu\in\hat{G}} O_\mu,
    \end{aligned}
\end{equation}
where $O_\mu\in \text{End}(\cH_\mu\otimes \mathbb{C}^{m_\mu})$ is an operator in the sub-sapce $\cH_\mu\otimes \mathbb{C}^{m_\mu}$. Furthermore, for a given pure state $\ket{\psi}$ in Hilbert space $\cH$, that can be decomposed into the following form:
\begin{equation}\label{eq:target_state_sym}
    \begin{aligned}
        \ket{\psi}=\bigoplus_{\mu\in \hat{G}}c_\mu\ket{\psi_\mu},
    \end{aligned}
\end{equation}
where $\ket{\psi_\mu}$ denotes the bipartite space in the space $\cQ_\mu\otimes \mathbb{C}^{m_\mu}$ and the amplitude $c_\mu\geq0$ satisfies $\sum_\mu c_\mu^2=1$. It means that the cost function can be written as a simplified form as follows:
\begin{equation}\label{eq:symmtry_cost}
    \begin{aligned}
        C:=\langle\psi|O|\psi\rangle=\sum_{\mu\in \hat{G}}c_\mu^2\langle \psi_\mu|O_\mu|\psi_\mu\rangle.
    \end{aligned}
\end{equation} 
Similar to the cost function in Eq.~\eqref{eq:loss_function}, suppose the coefficients $c_\mu$ and state $\ket{\psi_\mu}$ in irreducible space are parameterized via $\bm\alpha$ and $\bm \theta$. Then, we denote the above cost function as $C(\bm\alpha,\bm\theta)$, and have the following observation:

\begin{theorem}\label{thm:trade_off_sym}
    Let $H\subset \hat{G}$ be a subset of irrep labels. Then, we have $\mathbb{E}_{\bm \alpha,\bm \theta}[\partial C(\bm \alpha,\bm \theta)]=0$, and
    \begin{equation}
        \begin{aligned}
            {\rm Var}_{\bm \alpha,\bm \theta}[\partial C(\bm \alpha,\bm \theta)]\in\cO\left(\frac{{\rm Poly}(\log d_{\rm max})}{Ld_{\rm max}}\right),
        \end{aligned}
    \end{equation}
    where $d_{\rm max}:=\max_{\mu\in H} d_\mu m_\mu$ and $L:=|H|$.
\end{theorem}

The detailed proof can be found in Appendix~\ref{appendix:detailed_proof}. In particular, we consider the unitary group, i.e., $G=\textbf{SU}(2)$. Then the Hilbert space $\cH$ can be decomposed into a direct sum of invariant sub-space under the Schur-Weyl duality, that is,
\begin{equation}
    \begin{aligned}
        \cH\overset{\textbf{SU}(2)\times\textbf{S}_N}{\cong}\bigoplus_{\mu\in Y_N^2}\cQ_{\mu}\otimes\cP_\mu,
    \end{aligned}
\end{equation}
where $Y_N^2:=\{\lambda:=(\lambda_1,\lambda_2)|\lambda_1\geq\lambda_2\geq0\,\text{and}\,\sum_{i=1}^2\lambda_i=N\}$ denotes the set of Young Diagrams, $\cQ_\mu\otimes \cP_\mu$ is invariant sub-space labeled by $\mu$. The change of basis refers to schur transform $U_{\text{sch}}$. 
Since the irreps appearing in $\textbf{R}$ are both two-row Young diagrams, one can denote the Young diagram $\mu$ as a non-negative integer $j$, that is, $\mu:=\mu(j)=(N-j,j)$, where $j=0,1,\cdots,\lfloor \frac{N}{2}\rfloor$ then, we have 
\begin{equation}
    \begin{aligned}
       d_\mu=N-2j+1,\quad m_\mu=\frac{N!(N-2j+1)!}{(N-j+1)!j!(N-2j)!}.
    \end{aligned}
\end{equation}
Note that $d_\mu\in\cO(N)$, while some multiplicity $m_\mu$ can grow exponentially with the number of qubits~\cite{schatzki2024theoretical}. For instance, we have $m_\mu=\cO(4^N/N^2)$ when $j=N/2$. It means that to prevent BP, it is essential to avoid optimization within exponentially large subspaces. We denote the set of irrep labels by $H\subset\mathbb{Y}_N^2$, wherein each inequivalent irreducible representation space exhibits polynomial-level dimensionality.

In this case, suppose the observable $O$ can be decomposed into the partially irreducible spaces, i.e.,
\begin{equation}\label{eq:schur_decomposition}
    \begin{aligned}
        U_{\text{sch}}^\dagger OU_{\text{sch}}=\bigoplus_{\mu\in H} O_\mu,
    \end{aligned}
\end{equation}
where $O_\mu\in \text{End}(\cH_\mu\otimes \cP_\mu)$ is an operator in the sub-sapce $\cH_\mu\otimes \cP_\mu$, $U_{\text{sch}}$ denotes the schur transform, which can be implemented via a quantum circuit~\cite{bacon2005quantum,bacon2006efficient}. 

Therefore, it is straightforward to complete the gradient analysis based on the above assumption. 

\begin{corollary}\label{cor:tensor_su_2}
    Given an observable $O$ decomposing as specified in Eq.~\eqref{eq:schur_decomposition}, we have  $\mathbb{E}_{\bm \alpha,\bm \theta}[\partial C(\bm \alpha,\bm \theta)]=0$, and
    \begin{equation}
        \begin{aligned}
            {\rm Var}_{\bm \alpha,\bm \theta}[\partial C(\bm \alpha,\bm \theta)]\in\cO\left(\frac{{\rm Poly}(\log {\rm Ploy}(N))}{L{\rm Ploy}(N)}\right),
        \end{aligned}
    \end{equation}
    where $L:=|H|$.
\end{corollary}

Theorem~\ref{thm:trade_off_sym} and Corollary~\ref{cor:tensor_su_2} establish that under trainable group action parameterizations, the trainability of quantum neural networks (QNNs) is governed by the dimensions of their inequivalent irreducible representations. This justifies the standard assumption that measurements project predominantly onto irreducible subspaces of non-exponential dimension---a property analogous to the \textit{purity} condition discussed in BP issue~\cite{Alessandro2008,Fontana2024characterizing}

\begin{figure}[t]
    \centering
    \resizebox{0.4\textwidth}{!}{%
            \begin{quantikz}
                &\qwbundle{k}&& \gate{U(\bm\theta)} & \qw
        \end{quantikz}=\begin{quantikz}
            & \gate{U_3(\bm\theta_1)}\gategroup[4,steps=4,style={dashed,rounded corners,fill=orange!20, inner xsep=2pt},background,label style={label position=north east,anchor=south}]{$\times D$}& \ctrl{1}   & \qw & \targ{} & \qw \\
            & \gate{U_3(\bm\theta_2)} & \targ{} & \ctrl{2} & \qw& \qw  \\
            \wave&&&&&\\
            & \gate{U_3(\bm\theta_k)} & \qw & \targ{} & \ctrl{-3} & \qw 
        \end{quantikz}
    }
    \caption{The $k$-local unitary circuit, which includes single-qubit universal gate, $U_3$ and \text{CNOT} gates. Parameter $D$ is the depth of this circuit.}
    \label{fig:k-local circuit}
\end{figure}
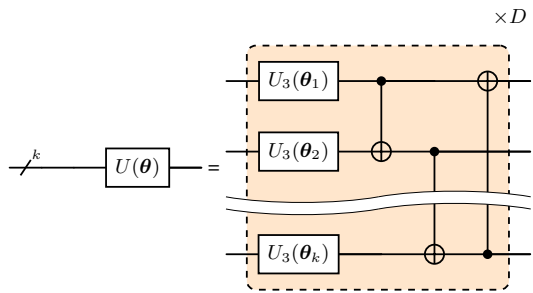

\section{Numerical Experiments}\label{sec:numerical_res}
In this section, we conduct numerical experiments to validate the trainability theory of the LCQNN framework and assess its capability to handle real-world data.

\begin{figure}[t]
    \centering
    \includegraphics[width=\linewidth]{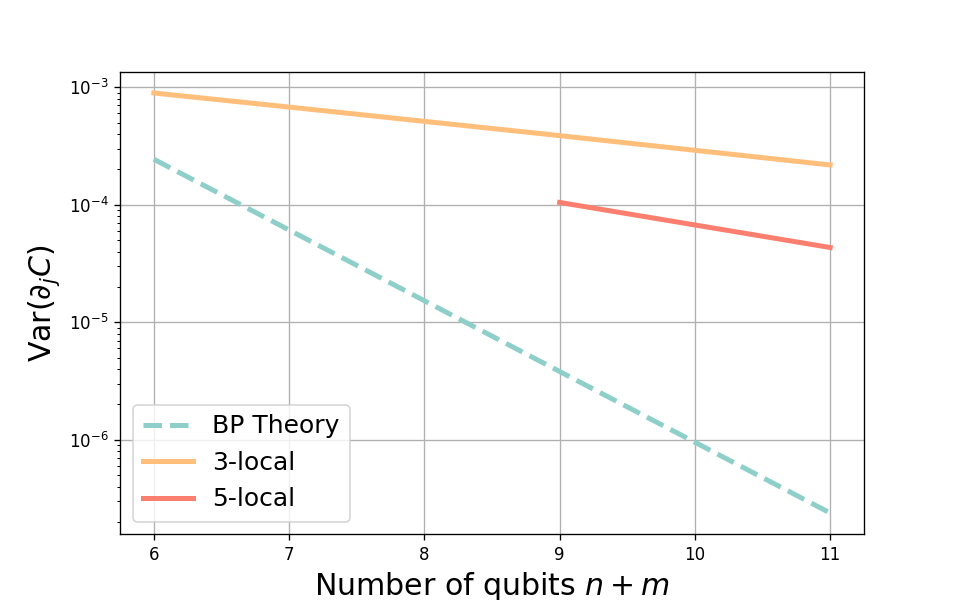}
    \caption{The trainability comparison between conventional QNNs and the LCQNN models. We set the number of systems for the coefficient layer, $m=3$, and $L=2^m$. As for the controlled systems, we set up $n\in\{3,4,6,8\}$ and the trainable circuit as shown in Fig.~\ref{fig:k-local circuit} with $D=3$ for $k=\{3,5\}$.}
    \label{fig:bp_exp_text_local}
\end{figure}

\begin{figure}[t]
    \centering
    \includegraphics[width=\linewidth]{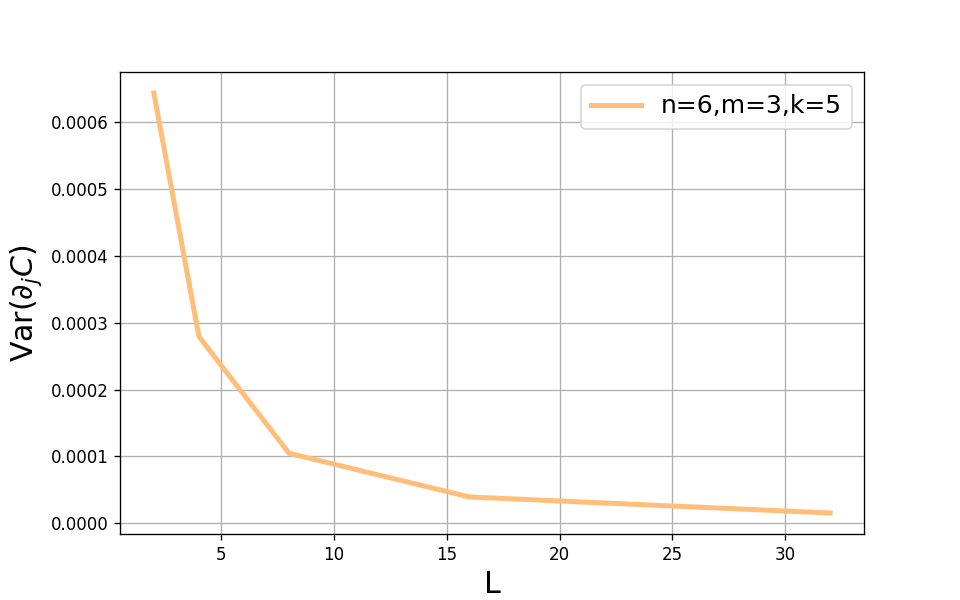}
    \caption{Variance of loss function with different layers $L$. We set the coefficient layer and controlled unitary with systems $m=3$ and $n=6$, respectively. We fix the size of the local unitary as $k=5$.}
    \label{fig:bp_exp_text_L}
\end{figure}

\subsection{Validation of Gradient Scaling}

This set of experiments aims to numerically validate our theoretical findings on the trainability of the LCQNN framework. The goal is to confirm that the gradient's variance scales with the size of the local unitaries ($k$), as predicted in Proposition~\ref{cor:k-local}, and exhibits an inverse relationship with the number of combined QNN blocks ($L$), as derived in Theorem~\ref{thm:trade-off_exp_bp}.

For the experimental setup, we implemented the LCQNN model with $k$-local unitaries. Each local unitary is constructed using the complex-entangled layer \cite{quairkit} shown in Fig.~\ref{fig:k-local circuit}. To verify the dependence on local size $k$, we configured an LCQNN with $m=3$ control qubits (allowing up to $L=8$ blocks) and a circuit depth of $D=3$. We then measured the gradient variance across systems with a growing number of total target qubits $n \in \{3, 4, 6, 8\}$, while fixing the local QNN sizes to $k=3$ and $k=5$. To validate the scaling with $L$, we fixed the architecture ($m=3$, $n=6$, $k=5$, $D=3$) and varied the number of active QNN blocks. For both experiments, a trainable parameter $\theta_j$ was selected, and the variance was averaged over 500 randomly uniform parameter initializations.

The results strongly support our theoretical predictions. As shown in Fig.~\ref{fig:bp_exp_text_local}, the gradient variance is determined by the local system size $k$ and remains essentially constant as the total system size $n$ increases, thereby successfully mitigating the barren plateau phenomenon. Concurrently, Fig.~\ref{fig:bp_exp_text_L} demonstrates a clear inverse relationship, showing that the variance of the cost function diminishes as the number of layers $L$ increases. These findings confirm that the LCQNN structure provides effective and tunable levers to ensure trainability.


\begin{table}[t]
\centering
\begin{tabular}{@{}ccc@{}}
\toprule
\# QNN Blocks (\(L\)) & QNN Depth (\(D\)) & Test Accuracy \\
\midrule
1 & 1 & 39.12\% $\pm$ 4.00\%   \\
2 & 1 & 45.60\%  $\pm$ 8.55\%  \\
4 & 1 & 47.07\% $\pm$ 5.42\% \\
\midrule
1 & 2 &  40.42\% $\pm$  7.06\%      \\
2 & 2 &  68.45\% $\pm$  3.24\% \\
4 & 2 &  68.83\%  $\pm$  2.85\%   \\
\midrule
1 & 4 &   41.15\%  $\pm$  2.41\%    \\
2 & 4 &  70.59\% $\pm$  3.46\% \\
4 & 4 &  70.98\% $\pm$  2.34\%   \\
\midrule
1 & 8 &  43.41\% $\pm$   2.43\%    \\
2 & 8 &  71.24\% $\pm$ 3.21\%  \\
4 & 8 &  72.36\% $\pm$   3.31\%  \\
\bottomrule
\end{tabular}
\caption{Test accuracy on the MNIST 4-class classification task (digits 0–3). The results were obtained using a 2-local LCQNN model configured with \(m=2\) control qubits and a 4-qubit working system (\(n=4\)). Performance is evaluated for different numbers of QNN blocks (\(L\)) and internal depths (\(D\)). Each value represents the mean accuracy ± standard deviation calculated over 5 independent experimental runs.}
\label{tab:mnist_classification}
\end{table}

\subsection{Quantum Classification on MNIST}
To evaluate the practical efficacy of the LCQNN framework, we applied it to a multi-class image classification task using the MNIST dataset. The objective is to demonstrate that the model can learn from real-world data and to analyze how its classification accuracy is impacted by its core architectural hyperparameters: the number of QNN blocks, $L$, and their internal depth, $D$.

Our experimental setup focused on classifying the first four digits (0-3) from the MNIST dataset. For data preprocessing, each classical image was resized to 4x4 pixels, flattened into a 16-element vector, and then normalized. This vector was embedded into the initial state of a 4-qubit system using amplitude encoding, serving as the input for the QNN. The LCQNN architecture was configured with $m=2$ control qubits, enabling the combination of up to $L=4$ QNN blocks, and a 4-qubit working system ($n=4$). Each QNN block is a \(k\)-local unitary with \(k=2\), adhering to the structure in Fig.~\ref{fig:k-local} and using the specific ansatz from Fig.~\ref{fig:k-local circuit}. For the classification output, the expectation value of the Pauli Z operator, \(\langle Z \rangle\), was measured on each of the four qubits. The resulting vector was used to compute a cross-entropy loss against the labels. All trainable parameters were uniformly initialized in the range \([0, 2\pi]\), and the model was trained for 2 epochs with a learning rate of 0.008.

The classification results, summarized in Table~\ref{tab:mnist_classification}, highlight the learning capabilities and architectural scalability of the LCQNN model. It is important to note that for a simpler binary classification version of this task, the LCQNN framework is capable of achieving high test accuracies in the 95-97\% range. However, the 4-class task was intentionally chosen for this study because it better demonstrates the advantage of increasing the number of QNN blocks ($L$). In a binary scenario that might rely on a single measurement output, the diverse features extracted by different blocks are not as easily distinguished, thus diminishing the observable benefit of a larger $L$. In contrast, the multi-output nature of the 4-class problem allows each block to contribute more distinctively to the final prediction.

With this context, the table reveals two distinct trends: for a fixed depth $D$, increasing the number of QNN blocks $L$ consistently improves the test accuracy, supporting the hypothesis that combining more feature extractors enhances performance. Likewise, for a fixed number of blocks $L$, increasing the depth $D$ also leads to better accuracy by boosting the expressivity of each individual block. These results confirm that the LCQNN can learn meaningful patterns from real-world data and that its performance can be systematically improved by scaling its architectural parameters, demonstrating its practical potential.

\section{Conclusion}\label{sec:conclusion}

In this work, we introduced the Linear Combination of Quantum Neural Networks (LCQNN) framework, which systematically addresses the trade-off between expressivity and trainability. By combining a coefficient layer with a control-unitary layer, LCQNN avoids excessive parameter space exploration while retaining a wide range of representable states, offering an explicit trade-off mechanism that can be leveraged in quantum architecture search.

We demonstrated that incorporating structural designs—such as $k$-local unitaries, selective circuit-depth expansions, and a tunable number of controlled unitaries—mitigates the barren plateau problem and reduces the risk of classical simulability, thereby providing a more tractable gradient landscape for optimization. Notably, the $k$-local unitaries function as feature filters, and as established in Proposition 2, this structure allows LCQNN to be scaled up significantly while preserving trainability. Our numerical experiments have verified these theoretical findings and confirmed the framework's practical applicability.

Furthermore, we extended the LCQNN framework to group action settings, where it operates as a linear combination of unitary actions from a finite or compact Lie group. In this context, we proved that trainability is governed by the maximum dimension of the irreducible representation space.

Future research could advance this work in several key directions. First, integrating advanced data embedding techniques from classical deep learning could enhance LCQNN’s adaptability for diverse tasks like image processing, natural language, and quantum chemistry. Second, developing hardware-friendly implementations, perhaps using restricted Lie-algebraic modules or low-rank factorization to minimize control gates, would improve performance on near-term quantum devices. Third, a deeper exploration of system-specific symmetries and group-theoretic embeddings could strengthen both the theoretical analysis of gradients and practical design principles. By continuing to refine the interplay between circuit structure and learning objectives, LCQNN has the potential to expand its reach across broader domains of quantum-enhanced machine learning.

\section{Acknowledgement}
This work was partially supported by the National Key R\&D Program of China (Grant No.~2024YFB4504004), the National Natural Science Foundation of China (Grant. No.~12447107), the Guangdong Provincial Quantum Science Strategic Initiative (Grant No.~GDZX2403008, GDZX2403001), the Guangdong Provincial Key Lab of Integrated Communication, Sensing and Computation for Ubiquitous Internet of Things (Grant No.~2023B1212010007), the Quantum Science Center of Guangdong-Hong Kong-Macao Greater Bay Area, and the Education Bureau of Guangzhou Municipality. G. Li was supported by National Natural Science Foundation of China (Grant No. 62402206).
\bibliographystyle{unsrt}
\bibliography{ref}

\onecolumngrid
\appendix
\setcounter{subsection}{0}
\setcounter{table}{0}
\setcounter{figure}{0}

\newpage

\begin{center}
\Large{\textbf{Appendix for} \\ \textbf{ LCQNN: Linear Combination of Quantum Neural Networks
}}
\end{center}

\renewcommand{\thetheorem}{S\arabic{theorem}}
\renewcommand{\theequation}{S\arabic{equation}}
\renewcommand{\theproposition}{S\arabic{proposition}}
\renewcommand{\thelemma}{S\arabic{lemma}}
\renewcommand{\thecorollary}{S\arabic{corollary}}
\renewcommand{\thefigure}{S\arabic{figure}}
\setcounter{equation}{0}
\setcounter{table}{0}
\setcounter{proposition}{0}
\setcounter{lemma}{0}
\setcounter{corollary}{0}
\setcounter{figure}{0}
\setcounter{figure}{0}

\section{Elements of Representation Theory}\label{sec:pre-rep}
\begin{definition}[Unitary Representation]
    Let $G$ be a group and $\cH$ a Hilbert space. The unitary representation of $G$ on $\cH$ is a group homomorphism $\mathbf{R}:g\to U_g$, where $U_g$ is unitary for all $g\in G$.
\end{definition}
Notice that a group action describes how a group $G$ permutes elements of a set $X$, while a group representation is a specific type of action in which $G$ acts linearly on a Hilbert space $\cH$. The group action mentioned in this work refers to the unitary representation of groups.

\begin{lemma}[Schur lemma]
    Let $U_g$ and $V_g$ be irreducible representations of the group $G$ on Hilbert spaces $\cH$ and $\cK$, respectively. Let $O:\cH\to\cK$ be an operator such that $OU_g=V_gO$ for all $g\in G$. If $U_g$ and $V_g$ are equivalent representations, then $O=\lambda I$, where $I$ is an identity operator; otherwise, $O=0$.
\end{lemma}

The Schur lemma is a building block for investigating the general form of an operator commuting with a group representation.

\begin{theorem}[Commutant]
    Let $U_g$ be a unitary representation of a group $G$ and $O$ be an operator satisfying that $[O,U_g]=0$ for all $g\in G$. Then, we have
    \begin{equation}\label{eq:iso_decompo}
        \begin{aligned}
            O=\bigoplus_{\mu \in \hat{G}} I_{\mu}\otimes O_\mu,
        \end{aligned}
    \end{equation}
    where $O_\mu$ denotes an operator on the multiplicity space of the irreducible representation $U_g^{\mu}$.
\end{theorem}
Therefore, it is straightforward to see that the group average operator $\overline{O}:=\frac{1}{|G|}\sum_{g\in G}U_g O U_g^\dagger$ can be decomposed into the form like Eq.~\eqref{eq:iso_decompo}.

The direct sum of irreducible representations leads to the decomposition of the Hilbert space $\cH$, i.e., isotypical decomposition. Specifically, suppose $G=\textbf{SU}(d)$ is a unitary group, and the carry space is $(\mathbb{C}^d)^{\otimes n}$. Then, we have
\begin{equation}
    \begin{aligned}
        (\mathbb{C}^d)^{\otimes n}=\bigoplus_{\mu\in Y_n^d}\cH_\mu\otimes \mathbb{C}^{m_\mu},
    \end{aligned}
\end{equation}
where $Y_n^d$ denotes the set of Young Diagrams. Therefore, for the case of $n=3$ and $d=3$, we have three irreducible representations, denoted by Young Diagrams, $\yng(3)$, $\yng(2,1)$ and $\yng(1,1,1)$, resepctively. The dimensions of irredcible representation spaces and the corresponding mutiplicity sapces are $d_{\yng(3)}=10$, $d_{\yng(2,1)}=8$, $d_{\yng(1,1,1)}=1$, $m_{\yng(3)}=1$, $m_{\yng(2,1)}=2$, and $m_{\yng(1,1,1)}=1$. Thus, we have 
\begin{equation}
    \begin{aligned}
        (\mathbb{C}^3)^{\otimes 3}=\cH_{\yng(3)}\oplus\cH_{\yng(2,1)}\oplus \cH_{\yng(2,1)}\oplus \cH_{\yng(1,1,1)}.
    \end{aligned}
\end{equation}

\section{Detailed Proofs}\label{appendix:detailed_proof}

\renewcommand\theproposition{\ref{cor:k-local}}
\setcounter{proposition}{\arabic{proposition}-1}
\begin{proposition}[k-local QNNs]\label{appendix:k-local}
     The LCQNN with $k$-local unitary is defined as $W:=\prod_{j=0}^{L-1}\left[{\rm C}\text{-}\left(\bigotimes_tU_{jt}\right)\right]\cdot(V\otimes I^{\otimes n})$, where $U_{jt}$ is a $k$- local unitary $2$-design, as shown in Fig.~\ref{fig:k-local}.
    Then, the LCQNN $W$ can achieve the pure state $\ket{\psi}:=\bigoplus_{j=0}^{2^m-1}\left(\sqrt{p_j}\ket{\psi_j}\bigotimes_i\ket{\phi_{jt}}\right)\in(\mathbb{C}^2)^{\otimes m+n}$, with $p_j=0$ for $j\geq L$. Furthermore, the variance of the cost gradient scales with the number of qubits as 
    \begin{equation}
        \begin{aligned}
            {\rm Var}_{\bm \alpha,\bm \theta}[\partial C(\bm \alpha,\bm \theta)]\in\cO\left(\frac{{\rm Poly}(k)}{L2^{k}}\right).
        \end{aligned}
    \end{equation}
\end{proposition}
\begin{proof}
     It is straightforward to see that the pure state $\ket{\psi}$ can be obtained after applying the LCQNN $W$, which characterizes the expressivity of this model.

    We mainly turn to investigate the trainability by focusing on these two terms $\mathbb{E}_{\bm \alpha,\bm \theta}[\partial C(\bm \alpha,\bm \theta)]$ and $\text{Var}_{\bm \alpha,\bm \theta}[\partial C(\bm \alpha,\bm \theta)]$. We further derive the cost function, which can be written in the following form:
    \begin{equation}\label{eq:loss_function_sym}
        \begin{aligned}
            C(\bm\alpha,\bm\theta):=\sum_{j=0}^{2^t-1}p_j(\bm\alpha)\sum_{i=1,\text{sys}_{n_i}\in\text{sys}_t}^{r}\bra{0}U_{jt}O_iU_{jt}^\dagger\ket{0},
        \end{aligned}
    \end{equation}
    where $O_i$ denotes the measurement on the $n_i$ system, and the identity operator has been omitted. The initial state $\ket{0}$ means the zero state on multi-systems. As shown in Fig.~\ref{fig:k-local}, denote the index $q_i$ as the number of local unitaries controlled by the $i$-th basis. The notation $\text{sys}_{n_i}\in\text{sys}_{t}$ indicates the systems $t\in[q_i]$ contains the system $n_i$. Due to $U_{jt}$ is $k$-local unitary, one can derive the BP theorem with respect to the $k$ systems.
    
    Similar to the proof in Theorem~\ref{thm:trade-off_exp_bp}, we obtain the target expectation and variance of the cost function, which is a standard BP theorem corresponding to the $k$-local system, which completes this proof.
\end{proof}

\renewcommand\thetheorem{\ref{thm:trade_off_sym}}
\setcounter{theorem}{\arabic{theorem}-1}
\begin{theorem}\label{appendix:trade_off_sym}
    Let $H\subset \hat{G}$ be a subset of irrep labels. Then, we have $\mathbb{E}_{\bm \alpha,\bm \theta}[\partial C(\bm \alpha,\bm \theta)]=0$, and
    \begin{equation}
        \begin{aligned}
            {\rm Var}_{\bm \alpha,\bm \theta}[\partial C(\bm \alpha,\bm \theta)]\in\cO\left(\frac{{\rm Poly}(\log d_{\rm max})}{Ld_{\rm max}}\right),
        \end{aligned}
    \end{equation}
    where $d_{\rm max}:=\max_{\mu\in H} d_\mu m_\mu$ and $L:=|H|$.
\end{theorem}
\begin{proof}
    Firstly, we are going to show that there exists a LCQNN that can achieve the target pure state. Denote $s:=\lceil \log d_{\rm max}\rceil$. Similar to Theorem~\ref{thm:trade-off_exp_bp}, we can obtain the following state by applying the parameterized circuit $V(\bm\alpha)$:
    \begin{equation}
        \begin{aligned}
            V(\bm\alpha)\ket{0}^{\otimes s}=\sum_{\mu\in H}c_\mu(\bm\alpha)\ket{j},
        \end{aligned}
    \end{equation}
    where $c_\mu(\bm\alpha):=\prod_{k=1}^t[\cos^2(\alpha_{\mu_k})]^{1-\mu_k}[\sin^2(\alpha_{\mu_k})]^{\mu_k}$, $L=2^t$ (w.l.o.g.), $\mu=\sum_{k=1}^t \mu_k 2^{t-k}$. Notice that to prevent the misuse of symbols, $\mu$ is employed not only to denote an irrep but also to indicate its position within the labels subset $H$.
    Trainable parameters $\alpha_{j_k}\in\bm\alpha:=\{\alpha_0,\cdots,\alpha_{L-2}\}$, and $\ket{j}$ is the basis on the Hilbert space $(\mathbb{C}^2)^{\otimes s}$. As a result, after applying the whole circuit $W$, one can obtain the output state as follows:
    \begin{equation}\label{eq:pqc_state_sym}
        \begin{aligned}
            W\ket{0}^{\otimes m+s}
            &=\bigoplus_{\mu\in H}c_\mu U_\mu(\bm \theta_\mu)\ket{0}^{\otimes s},
        \end{aligned}
    \end{equation}
    where $\bm \theta_j\in \bm \theta=\{\bm \theta_0,\cdots,\bm \theta_{L-1}\}$. Since $U_\mu(\bm\theta_\mu)$ is universal with some parameter $\bm\theta_\mu$ under the sub-space $\cQ_\mu\otimes\mathbb{C}^{m_\mu}$. Thus, there exist parameters $\bm\theta$ and $\bm\alpha$, such that the target pure state can be achieved. Notice that there exists a circuit that allows the system to revert to the Hilbert space $\cH$ by adjusting the trivial space.
    
    Secondly, we are going to investigate the trainability by focusing on these two terms $\mathbb{E}_{\bm \alpha,\bm \theta}[\partial C(\bm \alpha,\bm \theta)]$ and $\text{Var}_{\bm \alpha,\bm \theta}[\partial C(\bm \alpha,\bm \theta)]$. 
    One can find that the only difference from Theorem~\ref{thm:trade-off_exp_bp} is that
    $\mathbb{E}_{\bm\theta}[\partial^2_{\theta_\mu}\bra{0}^{\otimes n}U^\dagger_\mu(\bm\theta_\mu)OU(\bm\theta_\mu)\ket{0}^{\otimes n}]\in\cO(\text{Poly}(d_\mu m_\mu)/2^{(d_\mu m_\mu)})$. We choose the maximal dimension $d_{\rm max}$ of the irreducible space and its multiplicity space $\cQ_\mu\otimes \mathbb{C}^{m_{\mu}}$, which completes this proof.
\end{proof}

\end{document}